\documentclass[
reprint,
superscriptaddress,
nofootinbib,
amsmath,amssymb,
aps,
pra,
floatfix,longbibliography
]{revtex4-2}
\usepackage{graphicx}
\usepackage{dcolumn}
\usepackage[hypertexnames, breaklinks=true, bookmarksnumbered=true, bookmarksopen=true, colorlinks=true, linktocpage=true, citecolor=magenta, urlcolor=magenta, linkcolor=magenta]{hyperref}
\usepackage{amsmath,amssymb,amsfonts,amsthm}
\usepackage{mathtools} 
\usepackage[T1]{fontenc}
\usepackage[latin9]{inputenc}
\usepackage{txfonts}
\usepackage{bbm} 
\usepackage{bm} 
\usepackage{mathrsfs} 
\usepackage{xcolor}
\usepackage{natbib}
\allowdisplaybreaks

\newcommand{\RR}{\mathbbm{R}}

\DeclareMathOperator{\tr}{Tr}

\DeclareMathOperator{\real}{Re}
\DeclareMathOperator{\imag}{Im}

\newcommand{\id}{\mathbbm{1}}

\newtheorem{theorem}{Theorem}
\newtheorem{lemma}{\emph{Lemma}}
\newtheorem{proposition}{Proposition}

\newcommand{\bra}[1]{\langle #1|}
\newcommand{\ket}[1]{|#1 \rangle}
\newcommand{\braket}[2]{ \langle #1| #2 \rangle}
\newcommand{\ketbra}[2]{|#1 \rangle\! \langle#2|}
\newcommand{\norm}[1]{ \lVert #1  \rVert}

\newcommand{\I}{\mathscr{I}}
\newcommand{\RS}{\mathscr{R}}

\newcommand{\MI}{\hat{+}}
\newcommand{\MIi}{\hat{-}}

\begin{document}

\title{Resource theory of imaginarity: quantification and state conversion}
\author{Kang-Da Wu}
\affiliation{CAS Key Laboratory of Quantum Information, University of Science and Technology of China, \\ Hefei 230026, People's Republic of China}
\affiliation{CAS Center For Excellence in Quantum Information and Quantum Physics, University of Science and Technology of China, Hefei, 230026, People's Republic of China}

\author{Tulja Varun Kondra}
\affiliation{Centre for Quantum Optical Technologies, Centre of New Technologies, University of Warsaw, Banacha 2c, 02-097 Warsaw, Poland}

\author{Swapan Rana}
\affiliation{Centre for Quantum Optical Technologies, Centre of New Technologies,
University of Warsaw, Banacha 2c, 02-097 Warsaw, Poland}
\affiliation{S. N. Bose National Centre for Basic Sciences, JD Block, Sector III,
Kolkata 700106, India}
\affiliation{Physics \& Applied Mathematics Unit, Indian Statistical Institute, 203 B T Road,
Kolkata 700108, India}

\author{Carlo Maria Scandolo}
\affiliation{Department of Mathematics \& Statistics, University of Calgary, Calgary, AB, T2N 1N4, Canada}
\affiliation{Institute for Quantum Science and Technology, University of Calgary, Calgary, AB,  T2N 1N4, Canada}

\author{Guo-Yong Xiang}

\email{gyxiang@ustc.edu.cn}
\affiliation{CAS Key Laboratory of Quantum Information, University of Science and Technology of China, \\ Hefei 230026, People's Republic of China}
\affiliation{CAS Center For Excellence in Quantum Information and Quantum Physics, University of Science and Technology of China, Hefei, 230026, People's Republic of China}

\author{Chuan-Feng Li}

\affiliation{CAS Key Laboratory of Quantum Information, University of Science and Technology of China, \\ Hefei 230026, People's Republic of China}
\affiliation{CAS Center For Excellence in Quantum Information and Quantum Physics, University of Science and Technology of China, Hefei, 230026, People's Republic of China}

\author{Guang-Can Guo}

\affiliation{CAS Key Laboratory of Quantum Information, University of Science and Technology of China, \\ Hefei 230026, People's Republic of China}
\affiliation{CAS Center For Excellence in Quantum Information and Quantum Physics, University of Science and Technology of China, Hefei, 230026, People's Republic of China}

\author{Alexander Streltsov}
\email{a.streltsov@cent.uw.edu.pl}
\affiliation{Centre for Quantum Optical Technologies, Centre of New Technologies,
University of Warsaw, Banacha 2c, 02-097 Warsaw, Poland}

\begin{abstract} Complex numbers are widely used in both classical and quantum physics, and are indispensable components for describing quantum systems and their dynamical behavior. Recently, the resource theory of imaginarity has been introduced, allowing for a systematic study of complex numbers in quantum mechanics and quantum information theory. In this work we develop theoretical methods for the resource theory of imaginarity, motivated by recent progress within theories of entanglement and coherence. We investigate imaginarity quantification, focusing on the geometric imaginarity and the robustness of imaginarity, and apply these tools to the state conversion problem in imaginarity theory. Moreover, we analyze the complexity of real and general operations in optical experiments, focusing on the number of unfixed wave plates for their implementation. We also discuss the role of imaginarity for local state discrimination, proving that any pair of real orthogonal pure states can be discriminated via local real operations and classical communication. Our study reveals the significance of complex numbers in quantum physics, and proves that imaginarity is a resource in optical experiments.

\end{abstract}

\maketitle

\section{Introduction}

Quantum resource theories provide a unified approach for studying properties of quantum systems and their applications for quantum technology~\cite{Quantum-resource-1,Quantum-resource-2,ChitambarRevModPhys.91.025001}. The basis of any quantum resource theory is the definition of free states: these are states which are easy to prepare, compared to the effort for creating general quantum states. The concrete set of free states depends on the specific problem under study. As an example, in the resource theory of quantum entanglement~\cite{PhysRevLett.78.2275,Plenioquant-ph/0504163,HorodeckiRevModPhys.81.865} two remote parties can easily perform quantum operations in their local labs and can exchange classical information at no additional cost. The set of states which can be easily established in this setup is the set of separable states~\cite{WernerPhysRevA.40.4277}. Another important element of a resource theory is the definition of free operations, corresponding to transformations of a quantum system which are easy to implement. As with the free states, the concrete definition depends on the problem under study. In the theory of entanglement, the free operations are known as local operations and classical communication (LOCC)~\cite{BennettPhysRevA.54.3824}. 

Quantum entanglement is the main ingredient in many quantum technological tasks, such as quantum teleportation~\cite{BennettPhysRevLett.70.1895} and quantum key distribution~\cite{EkertPhysRevLett.67.661}. However, it has become clear in recent years that entanglement is not the sole source responsible for quantum advantages. This has led to the development of other quantum resource theories, in particular the resource theories of coherence~\cite{BaumgratzPhysRevLett.113.140401,StreltsovRevModPhys.89.041003}, quantum thermodynamics~\citep{delRio,Lostaglio-thermo}, purity~\cite{HorodeckiPhysRevA.67.062104,Gour20151,Streltsov_2018}, and asymmetry~\cite{Gour_2008,GourPhysRevA.80.012307,Marvian2014}. 

In this article, which is also a companion paper of~\cite{PRLversion},  we investigate the resource theory of imaginarity~\cite{Hickey+Gour.JPA.2018}, capturing the effort to create and manipulate quantum states with complex coefficients. The free states of this theory are \emph{real states}, i.e., quantum states with a real density matrix $\braket{m}{\rho|n}\in \RR$ for a fixed reference basis $\{\ket{m}\}$. The free operations are \emph{real operations}, i.e., quantum operations $\Lambda[\rho]=K_j \rho K_j^\dagger$ with real Kraus operators: $\braket{m}{K_j|n} \in \RR$. 

One of the main questions of any resource theory is whether two given states $\rho$ and $\sigma$ can be converted into each other via free operations. This is the \emph{state conversion problem}, and to solve it one needs to determine all quantum states $\sigma$ which can be obtained from a given state $\rho$. If the conversion $\rho \rightarrow \sigma$ is not possible deterministically, it might still be possible to achieve the transformation stochastically. The goal is then to determine the maximal probability of conversion $P(\rho \rightarrow \sigma)$. For the resource theory of imaginarity a full solution for probabilistic state conversion for all pure states has been recently announced in~\cite{PRLversion}, along with a complete solution for deterministic conversion for all single-qubit states. In this article we present a detailed discussion of these results, along with the technical proofs. Our methods make use of general properties of resource quantifiers and their connection to the state conversion problem. We present methods for imaginarity quantification, focusing on the geometric imaginarity and the robustness of imaginarity, and show that both measures have an operational interpretation via the state conversion problem. We also consider approximate imaginarity distillation, where the goal is to maximize the fidelity with the maximally imaginary states via real operations. Also for this problem the solution for all mixed states was announced in~\cite{PRLversion}, and we present a detailed proof in this article.

The sets of free states and free operations completely define a resource theory. However, for a resource theory to have an operational meaning it is also desirable that there exists an experimentally relevant setup where the free states are easy to prepare and the free operations are easy to implement, compared to general quantum states and operations. As we show in this work, the resource theory of imaginarity fulfills also this requirement when focusing on experiments with linear optics, if we restrict our elements to only wave plates and beam displacing devices. In standard path-encoded method, a pure quantum state $\ket{\psi}$ of dimension $d$ requires $2d-2$ unfixed wave plates to be created. In contrast, if all coefficients are real, we need only $d-1$ unfixed wave plates, due to the smaller number of parameters of the real vector space. Thus, restricting ourselves to quantum systems with real coefficients allows to reduce the number of unfixed wave plates, suggesting that real quantum states are easier to create, compared to general quantum states with complex elements. Starting from this observation, we analyze the effort for performing a general quantum operation in optical experiments, again focusing on the number of unfixed wave plates. As announced in the companion paper~\cite{PRLversion}, restricting ourselves to real operations allows to reduce the number of unfixed wave plates by $1/2$ in the limit of large dimension of the system under study, and a complete proof of this statement is presented in this work.

This article is structured as follows. In Section~\ref{sec:ResourceTheories} we discuss properties of general quantum resource theories. In Section~\ref{sec:Imaginarity} we define the resource theory of imaginarity and discuss its main features. In Section~\ref{sec:QuantifyingImaginarity} we present methods for imaginarity quantification. In Section~\ref{sec:StateTransformations} we consider state transformations via real operations.  Operational meaning of imaginarity in optical experiments is discussed in Section~\ref{sec:OpticalExperiments}. In Section~\ref{sec:StateDiscrimination} we discuss applications of imaginarity for local discrimination of quantum states. Conclusion is presented in Section~\ref{sec:Conclusion}.

\section{Quantum resource theories \\ and their features \label{sec:ResourceTheories}}
One of the main questions in any quantum resource theory is whether for two given quantum states $\rho$ and $\sigma$ there exists a free operation $\Lambda_f$ transforming $\rho$ into $\sigma$:
\begin{equation}
\sigma = \Lambda_f [\rho]. \label{eq:StateTransformation}
\end{equation}
The existence of such a transformation immediately implies that $\rho$ is more resourceful than $\sigma$, and in particular 
\begin{equation}
R(\rho)\geq R(\sigma) \label{eq:Monotonicity}
\end{equation}
 for any resource measure $R$.

If $\rho$ cannot be converted into $\sigma$ via a free operation, e.g. if $R(\rho)<R(\sigma)$, it might still be possible to achieve the conversion probabilistically, if the corresponding resource theory allows for stochastic free operations, with free Kraus operators $\{K_j\}$ such that $\sum_i K_j^\dagger K_j \leq \id$. It is further reasonable to assume that any incomplete set of free Kraus operators $\{K_j\}$ can be completed with free Kraus operators $\{L_i\}$ such that 
\begin{equation}
    \sum_i L_i^\dagger L_i + \sum_j K_j^\dagger K_j = \id. \label{eq:CompletenessStochastic}
\end{equation}
The maximal probability for converting $\rho$  into $\sigma$  is then defined as
\begin{equation}
P\left(\rho\rightarrow \sigma\right)=\max_{\{K_j\}} \left\{\sum_j p_j: \sigma=\frac{\sum_j K_j \rho K_j^\dagger}{\sum_j p_j}\right\} \label{eq:TransformationP}
\end{equation}
with probabilities $p_j=\tr[K_j\rho K_j^\dagger]$, and the maximum is taken over all (possibly incomplete) sets of free Kraus operators~$\{K_j\}$. The existence of a deterministic free operation between $\rho$ and  $\sigma$ as in Eq.~\eqref{eq:StateTransformation} is then equivalent to $P\left(\rho \rightarrow \sigma\right)=1$. 

If two states $\rho$ and $\sigma$ do not allow for deterministic neither stochastic transformations [i.e., $P\left(\rho \rightarrow \sigma\right)=0$], there remains the possibility to perform the transformation approximately. The figure of merit in this case is the maximal transformation fidelity
\begin{equation}
F(\rho\rightarrow\sigma)=\max_{\Lambda_f}\left\{ F(\Lambda_f[\rho],\sigma) \right\},
\end{equation}
with fidelity 
\begin{equation}
    F(\rho,\sigma)=\left(\tr\sqrt{\sqrt{\rho} \sigma \sqrt{\rho}}\right)^2, \label{eq:Fidelity}
\end{equation}
and the maximum is taken over all free operations $\Lambda_f$.

Any resource measure $R$ is monotonic under free operations, see Eq.~\eqref{eq:Monotonicity}. For resource theories which allow for stochastic conversion, one typically requires a stronger constraint on the resource measure, to be monotonic on average under free operations:
\begin{equation}
R(\rho)\geq \sum_j q_j R(\sigma_j). \label{eq:StrongMonotonicityR}
\end{equation}
Here, the states $\sigma_j$ arise from $\rho$ by applying a free operation: $\sigma_j = K_j \rho K_j^\dagger/q_j$ with free Kraus operators $K_j$, and $q_j$ is the corresponding probability: $q_j = \tr[K_j \rho K_j^\dagger]$. Quantifiers satisfying Eq.~\eqref{eq:StrongMonotonicityR} are also called \emph{strong resource monotones}. If $R$ is additionally convex, i.e., 
\begin{equation}
    R\left(\sum_j p_j \rho_j\right) \leq \sum_j p_j R(\rho_j),
\end{equation}
then strong monotonicity~\eqref{eq:StrongMonotonicityR} implies monotonicity~\eqref{eq:Monotonicity}.

A powerful upper bound on the conversion probability~\eqref{eq:TransformationP} can be obtained from any resource quantifier which is convex and strongly monotonic under free operations. For any such resource quantifier $R$, it holds:
\begin{equation}
    P\left(\rho \rightarrow \sigma\right) \leq \min\left\{\frac{R(\rho)}{R(\sigma)},1\right\}. \label{eq:Pbound}
\end{equation}
A proof of Eq.~(\ref{eq:Pbound}) was given in~\cite{Wu2020} for the resource theory of coherence, but the methods presented there can be applied for any resource theory. For this, let $\{K_j\}$ be a (possibly incomplete) set of free Kraus operators which transform $\rho$ to $\sigma$:
\begin{equation}
    \sigma = \frac{\sum_j K_j \rho K_j^\dagger}{\tr\left[\sum_j K_j \rho K_j^\dagger\right]}.
\end{equation}
We further assume that there exist free Kraus operators $\{L_i\}$ which complete the set $\{K_j\}$, see Eq.~(\ref{eq:CompletenessStochastic}). We now define the probabilities 
\begin{subequations}
\begin{align}
    p_j &= \tr[K_j \rho K_j^\dagger],\\
    q_i &= \tr[L_i \rho L_i^\dagger],
\end{align}
\end{subequations}
and post-measurement states
\begin{subequations}
\begin{align}
    \sigma_j &= \frac{K_j \rho K_j^\dagger}{p_j},\\
    \tau_i &= \frac{L_i \rho L_i^\dagger}{q_i}.
\end{align}
\end{subequations}
In general it holds that $P(\rho \rightarrow \sigma) \geq \sum_j p_j$, and there exists a set of free Kraus operators $\{K_j\}$ saturating this inequality. In the following, we assume that this is the case, i.e., $P(\rho \rightarrow \sigma) = \sum_j p_j$. Using convexity and strong monotonicity of $R$ it follows that
\begin{align}
    R(\rho) &\geq \sum_j p_j R(\sigma_j) + \sum_i q_i R(\tau_i) \geq \sum_j p_j R(\sigma_j) \nonumber \\
            &= P(\rho\rightarrow\sigma)\sum_j \frac{p_j}{P(\rho\rightarrow\sigma)}R(\sigma_j) \nonumber \\
            &\geq P(\rho\rightarrow\sigma) R(\sigma).
\end{align}
This completes the proof of Eq.~(\ref{eq:Pbound}).

An important resource quantifier is the robustness with respect to set of free states $\mathcal{F}$:
\begin{equation}
    R_\mathcal{F}(\rho) =\min_\tau\left\{s \geq 0:\frac{\rho+s\tau}{1+s}\in\mathcal{F}\right\}, \label{eq:Robustness}
\end{equation}
where the minimum is taken over all quantum states $\tau$ and all $s \geq 0$. The robustness measure has been first studied in the resource theory of entanglement~\cite{VidalPhysRevA.59.141,SteinerPhysRevA.67.054305,Plenioquant-ph/0504163} and more recently in the resource theory of coherence~\cite{NapoliPhysRevLett.116.150502,PianiPhysRevA.93.042107}. For any quantum resource theory, the robustness is closely related to the success probability in channel discrimination tasks~\cite{TakagiPhysRevLett.122.140402,Takagi+Regula.PRX.2019}. Here, a quantum channel $\Lambda_j$ is acting on a state $\rho$ with probability $p_j$. The goal of channel discrimination is to determine which channel has acted, by applying a measurement with POVM elements $\{M_j\}$ onto the final states $\Lambda_j(\rho)$. The success probability of the procedure is given as
\begin{equation}
    p_{\mathrm{succ}} (\rho,\{p_j,\Lambda_j\},\{M_j\}) = \sum_j p_j \tr[M_j\Lambda_j(\rho)].
\end{equation}
The connection between channel discrimination and the robustness measure is then given by~\cite{TakagiPhysRevLett.122.140402,Takagi+Regula.PRX.2019}
\begin{equation}
    \max_{\{p_j,\Lambda_j\},\{M_j\}} \frac{p_{\mathrm{succ}} (\rho,\{p_j,\Lambda_j\},\{M_j\})}{\max_{\sigma \in \mathcal{F}}p_{\mathrm{succ}} (\sigma,\{p_j,\Lambda_j\},\{M_j\})} = 1+ R_\mathcal{F}(\rho). \label{eq:RobustnessChannelDiscrimination}
\end{equation}
Eq.~\eqref{eq:RobustnessChannelDiscrimination} implies that for any resource state $\rho$ (i.e.\ a quantum state which is not element of $\mathcal{F}$) there exist a set of channels $\{\Lambda_j\}$ and a probability distribution $\{p_j\}$ such that the optimal guessing probability is strictly larger than for any $\sigma \in \mathcal F$.

\section{Resource theory of imaginarity \label{sec:Imaginarity}}

In the resource theory of imaginarity~\cite{Hickey+Gour.JPA.2018} the free states are called \emph{real states}, corresponding to the set of quantum states with a real density matrix: 
\begin{equation}
    \RS =\left\{\rho: \braket{m}{\rho|n} \in \RR\right\}.
\end{equation}
Similar to the resource theory of coherence~\cite{StreltsovRevModPhys.89.041003} this definition depends on the reference basis $\{\ket{m}\}$, and a state which is real with respect to one basis is not necessarily real with respect to another basis. Free operations in imaginarity theory are called \emph{real operations}, corresponding to quantum operations $\Lambda[\rho] = \sum_j K_j\rho K_j^\dagger$ with real Kraus operators~\cite{Hickey+Gour.JPA.2018}:
\begin{equation}
    \braket{m}{K_j|n} \in \RR. \label{eq:RealKraus}
\end{equation}
This definition implies that a real operation cannot create imaginarity, even if interpreted as a quantum measurement. It is clear from Eq.~(\ref{eq:RealKraus}) that an incomplete set of real Kraus operators $\{K_j\}$ can always be completed with real Kraus operators $\{L_i\}$ such that Eq.~(\ref{eq:CompletenessStochastic}) holds. In particular, we can define 
\begin{equation}
    L_0=\sqrt{\id - \sum_j K_j^\dagger K_j}.
\end{equation}
Since $K_j^\dagger K_j$ are real symmetric matrices, also $L_0$ is real, and moreover $L_0^\dagger L_0 + \sum_j K_j^\dagger K_j = \id$. An important state in the resource theory of imaginarity is the maximally imaginary state
\begin{equation}
    \ket{\MI} = \frac{1}{\sqrt{2}}(\ket{0} + i\ket{1}).
\end{equation}
The state $\ket{\MI}$ can be converted into any quantum state of arbitrary dimension via real operations~\cite{Hickey+Gour.JPA.2018}. The same holds true for the state $\ket{\MIi} = (\ket{0} - i\ket{1})/\sqrt{2}$.

We will now show that in the resource theory of imaginarity any pure state admits a simple generic form. Recalling that the definition of imaginarity is basis dependent, we define complex conjugation of a state $\ket{\psi}$ as follows:
\begin{equation}
    \ket{\psi^*} = \sum_j c_j^* \ket{j},
\end{equation}
where $\{\ket{j}\}$ is the reference basis, and $c_j = \braket{j}{\psi}$. The states $\ket{\psi}$ and $\ket{\psi^*}$ can also be expressed as
\begin{subequations} \label{eq:PureGamma}
\begin{align}
    \ket{\psi}   &= a\ket{\gamma_1} + ib\ket{\gamma_2},\\
    \ket{\psi^*} &= a\ket{\gamma_1} - ib\ket{\gamma_2},
\end{align}
\end{subequations}
where $a$ and $b$ are real numbers with $a^2+b^2=1$, and $\ket{\gamma_i}$ are real states. Equipped with these tools, we are now ready to prove the following proposition.
\begin{proposition} \label{prop:PureGeneric}
For any pure state $\ket{\psi}$ there exists a real orthogonal matrix $O$ such that
\begin{equation}
    O\ket{\psi} = \sqrt{\frac{1+|\braket{\psi^*}{\psi}|}{2}}\ket{0} + i \sqrt{\frac{1-|\braket{\psi^*}{\psi}|}{2}}\ket{1}. \label{eq:PureGeneric}
\end{equation}
\end{proposition}
\begin{proof}
In the first step, note that for any two real states $\ket{\gamma_1}$ and $\ket{\gamma_2}$ there exists a real orthogonal matrix $O$ such that
\begin{align}
    O\ket{\gamma_1} &= \ket{0}, \\
    O\ket{\gamma_2} &= \cos \theta \ket{0} + \sin \theta \ket{1},
\end{align}
where $\cos \theta = \braket{\gamma_1}{\gamma_2}$. Applying $O$ to the state $\ket{\psi}$ gives us
\begin{equation}
    O\ket{\psi} = (a+ib\cos \theta)\ket{0} + ib\sin\theta \ket{1}.
\end{equation}
Since the state $O\ket{\psi}$ is effectively a single-qubit state, we can associate a Bloch vector $\boldsymbol{r}$ with it, with coordinates
\begin{align}
    r_x = b^2 \sin (2\theta), \nonumber \\
    r_y = 2ab \sin (\theta), \\
    r_z = a^2 + b^2 \cos (2\theta). \nonumber
\end{align}

Let now $O'$ be a real orthogonal transformation, such that the Bloch vector $\boldsymbol{s}$ of the state $O'O\ket{\psi}$ is in the positive $y$-$z$ plane. Since $|\boldsymbol{s}|=1$, we can give the coordinates of $\boldsymbol{s}$ as follows:
\begin{align}
    s_x &= 0,\,\,\,\,\,s_y=|r_y|, \label{eq:sBloch} \\
    s_z &= \sqrt{1-r_y^2} = \sqrt{1 - 4a^2b^2 + 4a^2b^2 \cos^2 \theta}. \nonumber
\end{align}

From Eqs.~(\ref{eq:PureGamma}) we further obtain
\begin{subequations} \label{eq:ab}
\begin{align}
    a\ket{\gamma_1} = \frac{\ket{\psi} + \ket{\psi^*}}{2}, \\
    b\ket{\gamma_2} = \frac{\ket{\psi} - \ket{\psi^*}}{2i}.
\end{align}
\end{subequations}
These results allow us to express $a^2$ and $b^2$ as 
\begin{align}
    a^2 &= \left| \frac{\ket{\psi} + \ket{\psi^*}}{2} \right|^2 = \frac{1}{4}(2 + \braket{\psi^*}{\psi} + \braket{\psi}{\psi^*}), \\
    b^2 &= \left| \frac{\ket{\psi} - \ket{\psi^*}}{2i} \right|^2 = \frac{1}{4}(2 - \braket{\psi^*}{\psi} - \braket{\psi}{\psi^*}).
\end{align}
Recalling that $\cos\theta = \braket{\gamma_1}{\gamma_2}$ we arrive at
\begin{align}
    ab\cos \theta = ab\braket{\gamma_1}{\gamma_2} = \frac{1}{4i} (\braket{\psi^*}{\psi} - \braket{\psi}{\psi^*}).
\end{align}
Using these results, we can simplify Eqs.~(\ref{eq:sBloch}) as follows:
\begin{equation}
    s_x = 0, \,\,\,\,\, s_y =\sqrt{1 - |\braket{\psi^*}{\psi}|^2}, \,\,\,\,\, s_z = |\braket{\psi^*}{\psi}|.
\end{equation}
The pure state corresponding to the Bloch vector $\boldsymbol{s}$ is given by Eq.~(\ref{eq:PureGeneric}).
\end{proof}

As we will see in the next section, the generic form given in Proposition~\ref{prop:PureGeneric} is very relevant when it comes to imaginarity quantification for pure states. We will now prove that for any real state $\rho$ the fidelity with the maximally imaginary state $\ket{\MI}$ is bounded by $1/2$.

\begin{proposition} \label{prop:RSFidelity}
For any real state $\rho \in \RS$ it holds that 
\begin{equation}
    \braket{\MI}{\rho|\MI} \leq \frac{1}{2}, \label{eq:RSFidelity}
\end{equation}
with equality if $\rho$ is a single-qubit state.
\end{proposition}
\begin{proof}
For any real state we can write $\braket{\MI}{\rho|\MI}$ as follows:
\begin{align}
    \braket{\MI}{\rho|\MI} 
    &= \frac{(\bra{0} - i\bra{1})\rho (\ket{0} + i\ket{1})}{2}  \\
    &= \frac{\braket{0}{\rho|0} + \braket{1}{\rho|1}}{2}, \nonumber
\end{align}
where we used equality $\braket{0}{\rho|1} - \braket{1}{\rho|0} = 0$ which is true for any $\rho \in \RS$. This result directly implies Eq.~(\ref{eq:RSFidelity}), and it is clear that equality holds true for real qubit states.
\end{proof}

In the following, we will also make use of the fact that any $d$-dimensional quantum state $\rho$ can be decomposed as~\cite{Hickey+Gour.JPA.2018}
\begin{equation}
    \rho=\real (\rho) + i  \imag (\rho),
\end{equation}
where $\real(\rho) = \frac{1}{2}(\rho+\rho^T)$ is a real quantum state and $\imag(\rho) = \frac{1}{2i}(\rho-\rho^T)$ is a real anti-symmetric matrix. By spectral theorem $\imag(\rho)$ has an even rank $2r$ and there is a real orthogonal matrix $O$ such that $O \mathrm{Im}\,(\rho) O^T$ is block-diagonal \citep[p.~136]{horn_johnson_2012}:
\begin{equation}\label{eq:imaginarySpec}
O\,\mathrm{Im}\,\rho O^T=\bm{0}_{d-2r}\overset{r}{\underset{k=1}{\bigoplus}}\lambda_k
\begin{pmatrix}
 0&1 \\
 -1&0 
\end{pmatrix},
\end{equation}
where $\lambda_k>0$.

\section{Quantifying imaginarity \label{sec:QuantifyingImaginarity}}

Resource quantifiers for general quantum resource theories have been discussed in Section~\ref{sec:ResourceTheories}. For the resource theory of imaginarity, any imaginarity measure $\I$ should be zero on all real states:
\begin{equation}
    \I(\rho) = 0 \mathrm{\,\,\,\, for \, any \,\,\,\,} \rho \in \RS.
\end{equation}
Moreover, $\I$ should not increase under real operations:
\begin{equation}
    \I(\Lambda[\rho]) \leq \I(\rho) \label{eq:MonotonicityI}
\end{equation}
for any real operation $\Lambda$. A strong imaginarity monotone additionally fulfills 
\begin{equation}
    \sum_j q_j \I(\sigma_j) \leq \I(\rho) \label{eq:StrongMonotonicity}
\end{equation}
with $q_j = \tr[K_j \rho K_j^\dagger]$, $\sigma_j = K_j \rho K_j^\dagger/q_j$, and real Kraus operators $K_j$.

Monotonicity~(\ref{eq:MonotonicityI}) implies that all imaginarity measures are invariant under real orthogonal transformations:
\begin{equation}
    \I\left(O \rho O^T\right) = \I(\rho). \label{eq:MeasuresOrthogonal}
\end{equation}
To prove this, note that Eq.~(\ref{eq:MonotonicityI}) immediately implies 
\begin{equation}
    \I(O \rho O^T) \leq \I(\rho) \label{eq:ProofO-1}
\end{equation}
for all states $\rho$ and all real orthogonal matrices $O$. Defining $\sigma = O \rho O^T$ we further obtain 
\begin{equation}
    \I(\rho) = \I(O^T\sigma O) \leq \I(\sigma) = \I(O\rho O^T). \label{eq:ProofO-2}
\end{equation}
Combining Eqs.~(\ref{eq:ProofO-1}) and (\ref{eq:ProofO-2}) proves Eq.~(\ref{eq:MeasuresOrthogonal}).

For any pure state $\ket{\psi}$ any imaginarity measure $\I$ depends only on $|\braket{\psi^*}{\psi}|$:
\begin{equation}
    \I(\ket{\psi}) = f(|\braket{\psi^*}{\psi}|), \label{eq:PureMeasures}
\end{equation}
where the function $f$ depends on the concrete imaginarity measure $\I$. To see this, recall that for any pure state $\ket{\psi}$ there exists a real orthogonal matrix $O$ bringing $\ket{\psi}$ into the generic form~(\ref{eq:PureGeneric}). The proof of Eq.~(\ref{eq:PureMeasures}) is complete by using Eq.~(\ref{eq:MeasuresOrthogonal}) and noting that the generic form~(\ref{eq:PureGeneric}) depends only on $|\braket{\psi^*}{\psi}|$.

In the following, we will consider two concrete measures of imaginarity: geometric imaginarity and robustness of imaginarity.

\subsection{\label{sec:GeometricImaginarity} Geometric imaginarity}

For a pure state $\ket{\psi}$ we define the geometric imaginarity as
\begin{equation}
\I_g (\ket{\psi}) = 1 - \max_{\ket{\phi} \in \RS} |\braket{\phi}{\psi}|^2.
\end{equation}
For mixed states, we define $\I_g$ as the minimal average imaginarity, minimized over all decompositions of the state:
\begin{equation}
    \I_g(\rho)=\min \sum_j p_j \I_g (\ket{\psi_j}),
\end{equation}
where the minimum is taken over all ensembles $\{p_j,\ket{\psi_j}\}$ such that $\rho = \sum_j p_j \ket{\psi_j}\!\bra{\psi_j}$.
From the definition we see that $\I_g$ is convex: 
\begin{equation}
    \I_g\left(\sum_j p_j \rho_j\right)\leq\sum_j p_j \I_g(\rho_j).
\end{equation}
The definition of geometric imaginarity is analogous to fidelity-based quantifiers in other resource theories, in particular geometric measure of entanglement~\cite{Shimony1995,Barnum_2001,WeiPhysRevA.68.042307,Streltsov_2010} and geometric measure of coherence~\cite{streltsov2015measuring}. 

Due to Eq.~(\ref{eq:PureMeasures}), for pure states the geometric imaginarity must be a function of $|\braket{\psi^*}{\psi}|$. We will now go one step further and give an explicit expression for all pure states.

\begin{proposition} \label{prop:IgPure}
The geometric imaginarity of a pure state $\ket{\psi}$ is given as
\begin{equation}
    \I_g (\ket{\psi}) = \frac{1-|\braket{\psi^*}{\psi}|}{2}. \label{eq:IgAlpha}
\end{equation}
\end{proposition}
\begin{proof}
By Eq.~(\ref{eq:MeasuresOrthogonal}) and Proposition~\ref{prop:PureGeneric} it follows that 
\begin{equation}
    \I_g (\ket{\psi}) = \I_g\left(\sqrt{\frac{1+|\braket{\psi^*}{\psi}|}{2}}\ket{0} + i \sqrt{\frac{1-|\braket{\psi^*}{\psi}|}{2}}\ket{1}\right).
\end{equation}
To complete the proof, we will now evaluate $\I_g$ for any state of the form 
\begin{equation}
    \ket{\mu} = a_0\ket{0} + ia_1\ket{1}
\end{equation}
with $a_0 \geq a_1 \geq 0$ and $a_0^2 + a_1^2 = 1$. For any real state $\ket{\nu} = \sum_j b_j\ket{j}$ we have
\begin{equation}
    |\braket{\nu}{\mu}|^2 = |a_0b_0 + ia_1b_1|^2 = a_0^2 b_0^2 + a_1^2 b_1^2 \leq a_0^2,
\end{equation}
where the inequality follows from the fact that $\sum_j b_j^2 = 1$. Since $|\braket{0}{\mu}|^2=a_0^2$, we conclude that
\begin{equation}
    \max_{\ket{\nu} \in \RS} |\braket{\nu}{\mu}|^2 = a_0^2,
\end{equation}
and thus $\I_g(\ket{\mu}) = a_1^2$.
\end{proof}

Recall now that strong imaginarity monotones fulfill Eq.~(\ref{eq:StrongMonotonicity}). The following proposition shows that this is the case for geometric imaginarity.

\begin{proposition}
Geometric imaginarity is a strong imaginarity monotone.
\end{proposition}

\begin{proof}

To prove Eq.~\eqref{eq:StrongMonotonicity} in general, we will first prove it for pure states. By Proposition~\ref{prop:PureGeneric}, it is enough to prove it for states 
\begin{equation}
    \ket{\alpha} = \cos\alpha \ket{0} + i \sin\alpha \ket{1}
\end{equation}
with $\alpha \in [0,\pi/4]$, for which the geometric imaginarity is given by $\I_g(\ket{\alpha}) = \sin^2 \alpha$. For a pure initial state, all post-measurement states $\sigma_j$ are also pure. Thus, proving Eq.~(\ref{eq:StrongMonotonicity}) for pure states reduces to proving the inequality
\begin{equation}
    \sum_j\max_{\ket{\phi_j}\in\RS} |\braket{\phi_j}{K_j|\alpha}|^2 \geq \cos^2 \alpha, \label{eq:IgPureProof}
\end{equation}
where $\{K_j\}$ is a set of real Kraus operators. To prove Eq.~\eqref{eq:IgPureProof}, we first note that 
\begin{equation}
    \sum_j\max_{\ket{\phi_j}\in\RS} |\braket{\phi_j}{K_j|\alpha}|^2 \geq \sum_j \frac{|\braket{0}{K_j^T K_j|\alpha}|^2}{s_j},
\end{equation}
where we introduced 
\begin{equation}
    s_j = \braket{0}{K_j^T K_j|0}. \label{eq:si}
\end{equation}
 Recalling that all Kraus operators $K_j$ are real and using the explicit form of $\ket{\alpha}$ we obtain
\begin{align}
    |\braket{0}{K_j^T K_j|\alpha}|^2 &=  |\braket{0}{K_j^T K_j|0}|^2 \cos^2 \alpha \\
    &+  |\braket{0}{K_j^T K_j|1}|^2 \sin^2 \alpha \nonumber \\ 
    &\geq |\braket{0}{K_j^T K_j|0}|^2 \cos^2 \alpha, \nonumber
\end{align}
which further implies that 
\begin{equation}
    \sum_j\max_{\ket{\phi_j}\in\RS} |\braket{\phi_j}{K_j|\alpha}|^2 \geq \sum_j \frac{|\braket{0}{K_j^T K_j|0}|^2}{s_j} \cos^2 \alpha.
\end{equation}
Using the definition of $s_j$ in Eq.~\eqref{eq:si} and the fact that $\sum_j K_j^T K_j = \id$, we obtain the desired inequality~\eqref{eq:IgPureProof}.

The above arguments prove that $\I_g$ satisfies Eq.~\eqref{eq:StrongMonotonicity} when $\rho$ is pure. To extend this result to mixed states, consider an optimal decomposition of a mixed state $\rho = \sum_j p_j \ket{\psi_j}\!\bra{\psi_j}$, such that 
\begin{equation}
    \I_g(\rho)=\sum_j p_j \I_g (\ket{\psi_j}).
\end{equation}
Introducing the quantity $s_{jk}=\braket{\psi_k}{K_j^T K_j|\psi_k}$ we obtain
\begin{align}
    \sum_j q_j \I_g\left(\frac{K_j\rho K_j^T}{q_j}\right) 
    &= \sum_j q_j \I_g\left(\sum_k p_k \frac{K_j \ket{\psi_k}\!\bra{\psi_k} K_j^T}{q_j}\right) \nonumber \\
    &= \sum_j q_j \I_g\left(\sum_k \frac{p_k s_{jk}}{q_j} \times \frac{K_j \ket{\psi_k}\!\bra{\psi_k} K_j^T}{s_{jk}}\right) \nonumber \\
     &\leq \sum_{j,k} p_k s_{jk} \I_g\left( \frac{K_j \ket{\psi_k}\!\bra{\psi_k} K_j^T}{s_{jk}}\right) \nonumber \\
     &\leq \sum_j p_j \I_g (\ket{\psi_j}) = \I_g (\rho),
\end{align}
where in the first inequality we used the facts that $\I_g$ is convex. This completes the proof of Eq.~(\ref{eq:StrongMonotonicity}) for all mixed states.
\end{proof}

As we will show in Section~\ref{sec:StateTransformations}, geometric imaginarity for pure states admits an operational interpretation in the state conversion task.

\subsection{\label{Sec:Robustness} Robustness of imaginarity} 

For a general resource theory the robustness measure has been defined in Section~\ref{sec:ResourceTheories}. Following this approach, the robustness of imaginarity is defined as~\cite{Hickey+Gour.JPA.2018}
\begin{equation}
    \I_R(\rho)=\min_\tau\left\{s \geq 0:\frac{\rho+s\tau}{1+s}\in\RS\right\}, \label{eq:Robustness}
\end{equation}
where the minimum is taken over all quantum states $\tau$ and all $s \geq 0$. The following proposition gives a closed expression for the robustness of imaginarity of any quantum state $\rho$.

\begin{proposition} \label{prop:Robustness}
The robustness of imaginarity is equal to
\begin{equation}
    \I_R(\rho)=\frac{1}{2}\|\rho-\rho^T\|_1,
\end{equation}
where $T$ denotes transposition and $\norm{M}_1=\tr \sqrt{M^\dagger M}$ is the trace norm.
\end{proposition}

\begin{proof}
Let  $\tau^*$ be a quantum state achieving the minimum in Eq.~\eqref{eq:Robustness}. Then, the matrix $\rho+\I_R(\rho)\tau^*$ is real and Hermitian, and thus
\begin{equation}
\rho+\I_R(\rho)\tau^*=\rho^T+\I_R(\rho)(\tau^*)^T.
\end{equation}
We can now obtain a lower bound on the robustness of imaginarity as follows:
\begin{equation}
    \|\rho-\rho^T\|_1=\I_R(\rho)\|(\tau^*)^T-\tau^*\|_1\leq2\I_R(\rho),
\end{equation}
where we used the fact that $\|(\tau^*)^T-\tau^*\|_1\leq\norm{(\tau^*)^T}_1+\norm{\tau^*}_1= 2$. Thus, we have the bound
\begin{equation}
\I_R(\rho)\geq\frac{1}{2}\|\rho-\rho^T\|_1.
\end{equation}

To complete the proof, we will present a state  $\tau^*$ such that $\rho+s\tau^*$ is a real matrix with $s=\frac{1}{2}\|\rho-\rho^T\|_1$. To see this, recall that there exists a real orthogonal matrix $O$ such that $O\imag \rho O^T$ is block-diagonal as in Eq.~\eqref{eq:imaginarySpec} with coefficients $\lambda_m\geq 0$. If the dimension of the Hilbert space is even, we define $\tau^*$ to be a block-diagonal matrix of the form 
\begin{equation}
\tau^*=\frac{1}{2\sum_m\lambda_m}O^T
\begin{pmatrix}\lambda_{1} & -i\lambda_{1}\\
i\lambda_{1} & \lambda_{1}\\
 &  & \lambda_{2} & -i\lambda_{2}\\
 &  & i\lambda_{2} & \lambda_{2}\\
 &  &  &  & \ddots\\
 &  &  &  &  & \lambda_{k} & -i\lambda_{k}\\
 &  &  &  &  & i\lambda_{k} & \lambda_{k}
\end{pmatrix}
O.
\end{equation}
Note that $\rho+2\sum_m\lambda_m\tau^*$ is a real matrix, and moreover
\begin{equation}
\frac{1}{2}\|\rho-\rho^T\|_1=\|\imag \rho\|_1=2\sum_m \lambda_m.
\end{equation}This completes the proof for Hilbert space with even dimension. For odd dimension, the proof follows the same lines of reasoning, if we define the state $\tau^*$ as
\begin{equation}
\tau^*=\frac{1}{2\sum_m\lambda_m}O^T
\begin{pmatrix}\lambda_{1} & -i\lambda_{1} &  &  &  &  &  & 0\\
i\lambda_{1} & \lambda_{1} &  &  &  &  &  & 0\\
 &  & \lambda_{2} & -i\lambda_{2} &  &  &  & 0\\
 &  & i\lambda_{2} & \lambda_{2} &  &  &  & 0\\
 &  &  &  & \ddots &  &  & \vdots\\
 &  &  &  &  & \lambda_{k} & -i\lambda_{k} & 0\\
 &  &  &  &  & i\lambda_{k} & \lambda_{k} & 0\\
0 & 0 & 0 & 0 & \cdots & 0 & 0 & 0
\end{pmatrix}
O.
\end{equation}
This completes the proof of the proposition.
\end{proof}

Proposition~\ref{prop:Robustness} implies that the robustness of imaginarity coincides with the distance-based measure of imaginarity studied in~\cite{Hickey+Gour.JPA.2018}. For single-qubit states with Bloch vector $\boldsymbol{r}=(r_x,r_y,r_z)$ the robustness of imaginarity simplifies to 
\begin{equation}\label{eq:QubitRobustness}
    \I_R (\rho)=|r_y|.
\end{equation}
For pure states the robustness of imaginarity can be written as 
\begin{equation}\label{eq:RobustPhi}
    \I_R(\ket{\psi})=\sqrt{1-|\braket{\psi^*}{\psi}|^2}.
\end{equation}
This follows directly from Proposition~\ref{prop:Robustness} and the fact that $||\ket{\psi}\!\bra{\psi} - \ket{\phi}\!\bra{\phi}||_1 = 2\sqrt{1-|\braket{\psi}{\phi}|^2}$ holds true for any two pure states $\ket{\psi}$ and $\ket{\phi}$.

Due to Proposition~\ref{prop:Robustness}, the robustness of imaginarity has the property
\begin{equation}
    \I_R(p \rho_1 \oplus [1-p]\rho_2) = p \I_R (\rho_1) + (1-p)\I_R(\rho_2),
\end{equation}
which has previously been explored within the resource theory of quantum coherence~\cite{YuPhysRevA.94.060302}.

In the next section we will apply these results to quantum state conversion in imaginarity theory.

\section{State transformations via real operations \label{sec:StateTransformations}}

We will now discuss state transformations in the resource theory of imaginarity. Deterministic transformations for pure states have been considered in \cite{Hickey+Gour.JPA.2018}. The results presented in \cite{Hickey+Gour.JPA.2018} together with Proposition~\ref{prop:Robustness} imply that the conversion $\ket{\psi} \rightarrow \ket{\phi}$ is possible via real operations if and only if $\I_R(\ket{\psi})\geq \I_R(\ket{\phi})$. In the following, we will consider stochastic conversion for pure states.

\subsection{Stochastic transformations for pure states}

We now provide the maximal probability for converting a pure state $\ket{\psi}$ into another pure state $\ket{\phi}$ via real operations. The following theorem was announced in ~\cite{PRLversion}, and we give a full proof of it in this article.

\begin{theorem} \label{thm:PureConversion}
The maximum probability for a pure state transformation $\ket{\psi} \rightarrow \ket{\phi}$ via real operations is given by 
\begin{equation}
    P(\,\ket{\psi} \rightarrow \ket{\phi}\,) = \min \left\{\,\frac{1-|\,\braket{\psi^*\,}{\,\psi}\,|}{1-|\,\braket{\phi^*\,}{\,\phi}\,|}, 1\,\right\}.
\end{equation}
\end{theorem}

\begin{proof}

The proof will use properties of pure states within imaginarity theory (see Section~\ref{sec:Imaginarity}) and geometric imaginarity $\I_g$ (see Section~\ref{sec:GeometricImaginarity}). Since $\I_g$ is convex and a strong imaginarity monotone, the transition probability $P(\rho \rightarrow \sigma)$ is bounded as [see also Eq.~(\ref{eq:Pbound})]:
\begin{equation}
    P(\rho\rightarrow\sigma) \leq \frac{\I_g (\rho)}{\I_g(\sigma)}.
\end{equation}
In the case of pure states we can use Proposition~\ref{prop:IgPure} to obtain 
\begin{equation}
    P(\ket{\psi}\rightarrow \ket{\phi}) \leq \frac{1-|\braket{\psi^*}{\psi}|}{1-|\braket{\phi^*}{\phi}|}. \label{eq:StochasticPureBound}
\end{equation}

We will now consider the case 
\begin{equation}
    |\braket{\psi^*}{\psi}| \geq |\braket{\phi^*}{\phi}|, \label{eq:PureStochasticProof-1}
\end{equation}
and show that there exists a real operation saturating the bound~(\ref{eq:StochasticPureBound}). To see this, we first apply a real orthogonal transformation to the state $\ket{\psi}$, bringing it into the form
\begin{equation}
    \ket{\psi'} = \sqrt{\frac{1+|\braket{\psi^*}{\psi}|}{2}}\ket{0} + i \sqrt{\frac{1-|\braket{\psi^*}{\psi}|}{2}}\ket{1},
\end{equation}
see Proposition~\ref{prop:PureGeneric}. Then, we apply a real operation with the Kraus operators
\begin{equation}
    K_{0}=\left(\begin{array}{cc}
a & 0\\
0 & 1
\end{array}\right),\,\,\,\,\,K_{1}=\sqrt{\id-K_{0}^{2}},
\end{equation}
where $a$ is defined as 
\begin{equation}
    a = \sqrt{\frac{1-|\braket{\psi^*}{\psi}|}{1-|\braket{\phi^*}{\phi}|} \times \frac{1+|\braket{\phi^*}{\phi}|}{1+|\braket{\psi^*}{\psi}|}}.
\end{equation}
Note that $a \leq 1$ by Eq.~\eqref{eq:PureStochasticProof-1}. As it can be verified by inspection, the Kraus operator $K_0$ transforms $\ket{\psi'}$ into the state
\begin{equation}
    \ket{\phi'} = \sqrt{\frac{1+|\braket{\phi^*}{\phi}|}{2}}\ket{0} + i \sqrt{\frac{1-|\braket{\phi^*}{\phi}|}{2}}\ket{1},
\end{equation}
with probability 
\begin{equation}
    p = \frac{1-|\braket{\psi^*}{\psi}|}{1-|\braket{\phi^*}{\phi}|}. 
\end{equation}
Note that $\ket{\phi'}$ is equivalent to the desired state $\ket{\phi}$ up to a real orthogonal transformation, see Proposition~\ref{prop:PureGeneric}.

For the remaining case $|\braket{\psi^*}{\psi}| < |\braket{\phi^*}{\phi}|$, the transformation $\ket{\psi}\rightarrow \ket{\phi}$ can be achieved with unit probability~\cite{Hickey+Gour.JPA.2018}. 
\end{proof}

Theorem~\ref{thm:PureConversion} allows to determine the optimal probability for transitions between pure states in imaginarity theory. Moreover, it equips the geometric imaginarity $\I_g$ with an operational meaning: the maximal probability to convert $\ket{\psi}$ into $\ket{\phi}$ via real operations can be written as 
\begin{equation}
    P(\ket{\psi}\rightarrow \ket{\phi}) = \min\left\{\frac{\I_g(\ket{\psi})}{\I_g(\ket{\phi})},1\right\}.
\end{equation}

\subsection{Deterministic transformations for all single-qubit states}

So far we focused on transformations between pure states. We will now go one step further, and consider mixed states of a single qubit. Note that any single-qubit state can be represented by a real 3-dimensional Bloch vector. The following theorem provides a complete solution for the conversion problem via real operations for all qubit states. 

\begin{theorem} \label{thm:QubitTransitions}
For qubit states $\rho$ and $\sigma$ the transition $\rho \rightarrow \sigma$ is possible via real operations if and only if 
\begin{subequations}
\begin{align}
    s_{y}^{2} &\leq r_{y}^{2},\\
    \frac{1-s_{z}^{2}-s_{x}^{2}}{s_{y}^{2}} &\geq \frac{1-r_{z}^{2}-r_{x}^{2}}{r_{y}^{2}}, \label{eq:DeterministicQubit}
\end{align}
\end{subequations}
where $\boldsymbol{r}$ and $\boldsymbol{s}$ are the Bloch vectors of the initial and the target state, respectively.
\end{theorem}

This result was also announced in~\cite{PRLversion}, and a complete proof is presented in the following. For this, we will make use of methods developed earlier within the resource theory of quantum coherence~\cite{StreltsovPhysRevLett.119.140402,ChitambarPhysRevLett.117.030401,ChitambarPhysRevA.94.052336}. To use this analogy in an optimal way, we introduce a new set of operations, which we term \textit{$y$-$z$-preserving operations} and denote them by $\Lambda_{yz}$. They correspond to single-qubit quantum operations which map the $y$-$z$ plane of the Bloch space onto itself, i.e., if a state $\rho$ has a Bloch vector in the $y$-$z$ plane, then $\Lambda_{yz}[\rho]$ also has this property. In the same way, $x$-$z$-preserving operations map the set of real states onto itself. Similarly, $z$-preserving operations map diagonal states onto diagonal states, thus corresponding to maximally incoherent operations (MIO)~\cite{Aberg2006}. 

In the following we will prove two lemmas which will be useful for proving Theorem~\ref{thm:QubitTransitions}, and which also demonstrate a close relation between the resource theories of coherence and imaginarity.

\begin{lemma} \label{lem:QubitTransformations-1}
Let $\rho_{\mathrm{r}}$ and $\sigma_{\mathrm{r}}$ be qubit states
with Bloch vectors in the $x$-$z$ plane. If there exists a $y$-$z$-preserving operation $\Lambda_{yz}$ such
that $\Lambda_{yz}[\rho_\mathrm{r}]=\sigma_\mathrm{r}$,  there also exists a $z$-preserving operation $\Lambda_z$ such
that $\Lambda_z[\rho_\mathrm{r}]=\sigma_\mathrm{r}$.
\end{lemma}
\begin{proof}
Since $\Lambda_{yz}$ is $y$-$z$
preserving, it converts both states $\ket{0}$ and $\ket{1}$ into
states $\mu_{0}$ and $\mu_{1}$ with Bloch vectors in the $y$-$z$
plane, i.e., $\mu_{0}$ and $\mu_{1}$ have purely imaginary off-diagonal
elements. This implies that any convex combination of $\ket{0}$ and
$\ket{1}$ is also converted into a state with purely imaginary off-diagonal
elements.

Let now $\{K_{j}\}$ be the Kraus operators of $\Lambda_{yz}$. We introduce
another transformation 
\begin{equation}
\Lambda'(\rho)=\sum_{j}L_{j}\rho L_{j}^{\dagger}
\end{equation}
 with Kraus operators $L_{j}=K_{j}^{*}$. It is straightforward to
verify that $\{L_{j}\}$ is indeed a valid set of Kraus operators:
\begin{equation}
\sum_{j}L_{j}^{\dagger}L_{j}=\sum_{j}\left(K_{j}^{*}\right)^{\dagger}K_{j}^{*}=\sum_{j}K_{j}^{T}K_{j}^{*}=\sum_{j}\left(K_{j}^{\dagger}K_{j}\right)^{*}=\id.
\end{equation}
Moreover, when applied to any state $\tau_{\mathrm{r}}$ in the
$x$-$z$ plane, we obtain 
\begin{equation}
\Lambda'(\tau_{\mathrm{r}})=\sum_{j}K_{j}^{*}\tau_{\mathrm{r}}K_{j}^{T}=\left(\sum_{j}K_{j}\tau_{\mathrm{r}}K_{j}^{\dagger}\right)^{*}=\left[\Lambda_{yz}(\tau_{\mathrm{r}})\right]^{T},
\end{equation}
where in the last step we used the fact that $\Lambda_{yz}(\tau_{\mathrm{r}})$
is Hermitian. It follows that \begin{subequations}
\begin{align}
\Lambda'[\rho_{\mathrm{r}}] & =\sigma_{\mathrm{r}},\\
\Lambda'[\ket{0}\!\bra{0}] & =\mu_{0}^{T},\\
\Lambda'[\ket{1}\!\bra{1}] & =\mu_{1}^{T}.
\end{align}
\end{subequations}

In the next step, we introduce the transformation 
\begin{equation}
\tilde{\Lambda}(\rho)=\frac{1}{2}\Lambda_{yz}(\rho)+\frac{1}{2}\Lambda'(\rho).
\end{equation}
Recalling that the states $\mu_{0}$ and $\mu_{1}$ have purely imaginary
off-diagonal elements, we further obtain 
\begin{subequations}
\begin{align}
\tilde{\Lambda}[\rho_{\mathrm{r}}] & =\sigma_{\mathrm{r}},\\
\tilde{\Lambda}[\ket{0}\!\bra{0}] & =\frac{1}{2}\left(\mu_{0}+\mu_{0}^{T}\right)=\sum_{j}\braket{j}{\mu_{0}|j}\ket{j}\!\bra{j},\\
\tilde{\Lambda}[\ket{1}\!\bra{1}] & =\frac{1}{2}\left(\mu_{1}+\mu_{1}^{T}\right)=\sum_{j}\braket{j}{\mu_{1}|j}\ket{j}\!\bra{j}.
\end{align}
\end{subequations}
This implies that $\tilde{\Lambda}$ is a $z$-preserving operation
transforming $\rho_{\mathrm{r}}$ onto $\sigma_{\mathrm{r}}$. 
\end{proof}

In the next step, we will use Lemma~\ref{lem:QubitTransformations-1} to characterize the set of real states achievable from a given real state $\rho$ via $y$-$z$-preserving operations.

\begin{lemma} \label{lem:QubitTransformations-2}
Let $\rho_\mathrm r$ and $\sigma_\mathrm r$ be qubit states in the $x$-$z$ plane of
the Bloch sphere. Then, there exists a $y$-$z$-preserving operation
such that $\sigma_\mathrm r = \Lambda_{yz}[\rho_\mathrm r]$ if and only if 
\begin{subequations}
\begin{align}
s_{x}^{2} & \leq r_{x}^{2},\label{eq:LemmaQubit-1}\\
\frac{1-s_{z}^{2}}{s_{x}^{2}} & \geq\frac{1-r_{z}^{2}}{r_{x}^{2}},\label{eq:LemmaQubit-2}
\end{align}
\end{subequations}
where $\boldsymbol{r}$ and $\boldsymbol{s}$ denote the Bloch vectors of $\rho_\mathrm r$ and $\sigma_\mathrm r$, respectively.
\end{lemma}
\begin{proof}
We will first prove that a $y$-$z$-preserving operation violating
Eq.~(\ref{eq:LemmaQubit-1}) and/or Eq.~\eqref{eq:LemmaQubit-2} does not
exist. Assume -- by contradiction -- that there exists a $y$-$z$-preserving operation violating Eq.~(\ref{eq:LemmaQubit-1}) and/or Eq.~\eqref{eq:LemmaQubit-2}. Then, by Lemma~\ref{lem:QubitTransformations-1} there must also exist a $z$-preserving (i.e. MIO) operation such that $\sigma_\mathrm r = \Lambda_\mathrm{MIO}[\rho_\mathrm r]$. Such a transformation does not exist due to results in \cite{StreltsovPhysRevLett.119.140402,ChitambarPhysRevLett.117.030401,ChitambarPhysRevA.94.052336}.

We will now show that a $y$-$z$-preserving operation exists if Eqs.~\eqref{eq:LemmaQubit-1} and \eqref{eq:LemmaQubit-2} are fulfilled. Note that $\sigma_z$ and any rotation around the $x$-axis are $y$-$z$-preserving operations. Thus, we can restrict ourselves to the positive part of the Bloch space, i.e., all Bloch coordinates considered in the following are non-negative. Moreover, we are interested in the boundary of the achievable region, characterized by the maximal $s_x$ for a given $s_z$. 

If $s_{z}>r_{z}$, Eq.~\eqref{eq:LemmaQubit-2} guarantees that Eq.~\eqref{eq:LemmaQubit-1} is satisfied. A $y$-$z$-preserving
operation fulfilling Eq.~\eqref{eq:LemmaQubit-2} with equality
is given by the Kraus operators
\begin{equation}
K_{1}=\left(\begin{array}{cc}
a_{1} & 0\\
0 & b_{1}
\end{array}\right),\,\,\,\,\,K_{2}=\left(\begin{array}{cc}
0 & b_{2}\\
a_{2} & 0
\end{array}\right),
\end{equation}
where the parameters $a_{i}$ and $b_{i}$ are chosen as \begin{subequations}
\begin{align}
a_{1} & =\cos\frac{\theta-\nu}{2},\,\,\,\,a_{2}=\sin\frac{\theta-\nu}{2},\\
b_{1} & =\sin\frac{\theta+\nu}{2},\,\,\,\,b_{2}=\cos\frac{\theta+\nu}{2},
\end{align}
\end{subequations}with $\nu=\arctan[r_{z}\tan\theta]$ and parameter
$\theta$ is in the range $[0,\frac{\pi}{2}]$. By varying $\theta$ it is possible to attain any value for $s_z$ in the range $[r_z,1]$. This proves that for $s_{z}>r_{z}$ the boundary of the achievable region is characterized by Eq.~(\ref{eq:LemmaQubit-2}).

For $s_z \leq r_z$ Eq.~(\ref{eq:LemmaQubit-1}) ensures that Eq.~(\ref{eq:LemmaQubit-2}) is fulfilled. The boundary of the achievable region is then obtained by the $y$-$z$-preserving operation
\begin{equation}
    \Lambda[\rho]= (1-p) \rho + p \sigma_x \rho \sigma_x
\end{equation}
with $p$ in the range $[0,1/2]$. This proves that for $s_z \leq r_z$ the boundary of the achievable region is determined by Eq.~(\ref{eq:LemmaQubit-1}).
\end{proof}

Equipped with these results, we are now ready to prove Theorem~\ref{thm:QubitTransitions}. Since rotations around the $y$-axis correspond to real unitaries, we can without loss of generality assume that the initial and the final state have Bloch vectors in the $y$-$z$ plane. It is thus enough to prove the statement for
\begin{subequations}
\begin{align}
    s_y^2 &\leq r_y^2 \\
    \frac{1-s_{z}^{2}}{s_{y}^{2}} &\geq\frac{1-r_{z}^{2}}{r_{y}^{2}}.
\end{align}
\end{subequations}
The proof of the theorem now directly follows from Lemma~\ref{lem:QubitTransformations-2} by symmetry, exchanging the $x$ and $y$ directions.

In Fig.~\ref{fig:QubitTransitions} we show the $y$-$z$-projection of the accessible region for three different initial states. The complete region can be obtained by rotation around the $y$-axis.

\begin{figure}
	\centering
	\includegraphics[width=0.9\columnwidth]{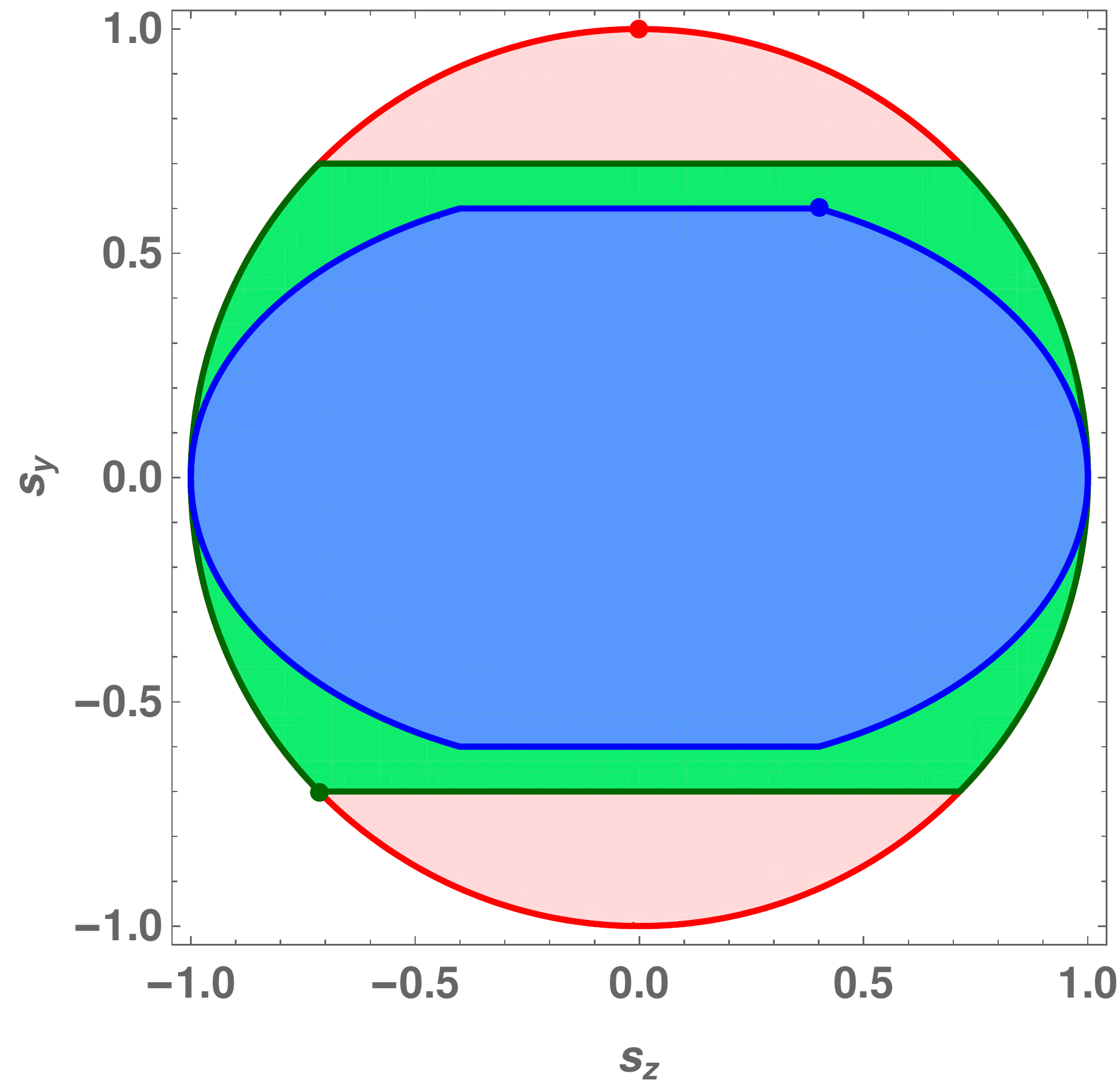}
	\caption{\label{fig:QubitTransitions} \textbf{State transformation via real operations for qubit systems.} The plot shows the $y$-$z$ projections of accessible states for initial qubit states with Bloch vectors $(0,0.6,0.4)$ [blue dot], $(0,-0.7,-\sqrt{0.51})$ [green dot], and $(0,1,0)$ [red dot]. Note that the second and third states are pure. The corresponding accessible area in the $y$-$z$ plane is shown in blue, green, and red, respectively. The full accessible area is obtained by rotation around the $y$-axis.
	}
\end{figure}

\subsection{Approximate imaginarity distillation}

As discussed in Section~\ref{sec:ResourceTheories}, it is always possible to perform approximate state transformations, even if neither deterministic nor stochastic conversion is possible. For an initial state $\rho$ we are then interested in the maximal fidelity between $\Lambda[\rho]$ and the target state $\sigma$, maximized over all real operations $\Lambda$:
\begin{equation}
F(\rho\rightarrow\sigma)=\max_{\Lambda}\left\{ F(\Lambda[\rho],\sigma) \right\},
\end{equation}
where the fidelity $F(\rho,\sigma)$ is defined in Eq.~(\ref{eq:Fidelity}). If the target state is the maximally imaginary state $\ket{\MI}$ the corresponding quantity is called \emph{fidelity of imaginarity}~\cite{PRLversion}:
\begin{equation}
    F_\mathrm{I}(\rho) = F\left(\rho \rightarrow \ket{\MI}\!\bra{\MI}\right).
\end{equation}
The following theorem gives a closed expression for the fidelity of imaginarity for any quantum state. This result was announced in~\cite{PRLversion}, and a complete proof is presented below.

\begin{theorem} \label{thm:FI}
For any quantum state $\rho$ the fidelity of imaginarity is given as 
\begin{equation}\label{maxfid}
    F_\mathrm{I}(\rho)=\frac{1+\I_R(\rho)}{2} = \frac{1}{2} + \frac{1}{4}||\rho - \rho^T||_1,
\end{equation}
where $T$ denotes transposition and $||M||_1=\tr \sqrt{M^\dagger M}$ is the trace norm.
\end{theorem}

\begin{proof}

From the definition of robustness of imaginarity $\I_\mathrm{R}$ (see Section~\ref{Sec:Robustness}), we can write $\rho$ as
\begin{equation}
\rho=[1+\I_R(\rho)]\delta-\I_R(\rho)\tau,
\end{equation}
with some quantum state $\tau$ and a real state $\delta$. By applying a real operation $\Lambda$ on both sides we obtain
\begin{equation}
\bra{\hat{+}}\Lambda(\rho)\ket{\hat{+}}=[1+\I_R(\rho)]\bra{\hat{+}}\Lambda(\delta)\ket{\hat{+}}-\I_R(\rho)\bra{\hat{+}}\Lambda(\tau)\ket{\hat{+}}.
\end{equation}
Since $\Lambda$ is a real operation, we have $\Lambda(\delta)\in\RS$. Applying Proposition~\ref{prop:RSFidelity} we obtain
\begin{equation}
\bra{\hat{+}}\Lambda(\delta)\ket{\hat{+}}\leq\frac{1}{2},
\end{equation}
which proves the bound
\begin{equation}
\bra{\hat{+}}\Lambda(\rho)\ket{\hat{+}}\leq \frac{1}{2}[1+\I_R(\rho)]. \label{eq:RobustnessBound}
\end{equation}

We will now show that this bound is achievable by a real operation $\Lambda$. If the dimension is even, we define $\Lambda$ via the real Kraus operators
\begin{equation}
    K_m=\ketbra{1}{2m}+\ketbra{0}{2m+1},\, m=0,1,\ldots, d/2-1.
\end{equation}
For odd dimension, the Kraus operators $K_m$ are defined in the same way for $m\leq \lfloor d/2 \rfloor-1$, and we further define 
\begin{equation}
    K_{\lfloor d/2 \rfloor} = \ket{0}\!\bra{d-1}.
\end{equation}

Let now $O$ be a real orthogonal matrix such that $O \imag(\rho) O^T$ is block-diagonal as in Eq.~(\ref{eq:imaginarySpec}). We see that $\Lambda[O\real (\rho)O^T]$ is a real single-qubit state, which by Proposition~\ref{prop:RSFidelity} implies that 
\begin{equation}
    \braket{\MI}{\Lambda[O\real(\rho)O^T]|\MI}=\frac{1}{2}.
\end{equation}
Moreover, we have 
\begin{equation}
\Lambda[O\mathrm{Im}(\rho)O^T]=\left(\sum_{m=0}^{\lfloor d/2 \rfloor-1}\lambda_m\right)\left(\ketbra{1}{0}-\ketbra{0}{1}\right).
\end{equation}
The fidelity of the final state with the maximally imaginary state can now be evaluated as follows:
\begin{align}
\braket{\MI}{\Lambda[O\rho O^T]|\MI} &= \braket{\MI}{\Lambda[O\real(\rho) O^T]|\MI} \label{eq:FIproof} \\
&+ i \braket{\MI}{\Lambda[O\imag(\rho) O^T]|\MI} \nonumber \\
&= \frac{1}{2}(1+2\sum_m\lambda_m). \nonumber
\end{align}
To complete the proof, note that the robustness of imaginarity can be expressed as (see Section~\ref{Sec:Robustness})
\begin{equation}
    \I_R(\rho) = \frac{1}{2}||\rho-\rho^T||_1 = ||\imag(\rho)||_1 = 2\sum_m\lambda_m.
\end{equation}
Using this result in Eq.~\eqref{eq:FIproof}, we obtain
\begin{equation}
    \braket{\MI}{\Lambda[O\rho O^T]|\MI} = \frac{1}{2}[1+\I_R(\rho)].
\end{equation}

In summary, we proved that the fidelity $\bra{\hat{+}}\Lambda(\rho)\ket{\hat{+}}$ is upper-bounded by Eq.~\eqref{eq:RobustnessBound}, and that this upper bound is achievable for any state $\rho$ with a suitably chosen real operation $\Lambda$. This completes the proof of the theorem.
\end{proof}

The relation between the robustness measure and conversion fidelity -- as given in Theorem~\ref{thm:FI} -- can also be extended to a general class of resource theories~\cite{RegulaPhysRevA.101.062315}.

\section{Imaginarity as a resource in optical experiments \label{sec:OpticalExperiments}}

In this section, we show that imaginarity can be regarded as a resource in linear optical experiments. We focus on optical setups with the following assumptions: (1) The quantum information is encoded in polarization or path degrees of freedom. (2) The optical elements we can use are limited to standard linear optical elements, i.e., half(quarter) wave plates and beamsplitters.

Under above assumptions, real operations can be implemented more economically, compared to general quantum operations. We begin with a simple observation, that implementing a general unitary on photon polarization requires to control at least $3$ wave plates (this is due to a qubit unitary being specified by $3$ parameters), whereas only one half-wave plate is needed if the unitary has only real components, e.g., rotations with respect to the $y$-axis. When restricting the optical elements to half (quarter)-wave plates, a rotation about the $z$-axis needs two additional quarter-wave plates compared to a rotation about the $y$-axis. This observation is the first evidence that the set of real operations is potentially easier to implement in terms of the number of optical elements, compared to the set of complex quantum operations. 

We then consider single-qubit measurement with $n$ outcomes. As illustrated in Fig.~\ref{fig:QubitMeasurement}, any such measurement can be implemented with $8n-5$ unset wave plates. If $n=2$, we have two Kraus operators $K_0$ and $K_1$ with
\begin{equation}
    K_1^\dagger K_1 = \id - K_0^\dagger K_0. \label{eq:QubitMeasurement}
\end{equation}
By singular value decomposition, there are unitaries $U_i$ and $V_i$ such that $K_j = U_j S_j V_j$, and $S_j$ are diagonal matrices with nonnegative entries. By Eq.~\eqref{eq:QubitMeasurement} we obtain 
\begin{equation}
    V_1^\dagger S_1^2 V_1 = V_0^\dagger (\id -  S_0^2) V_0,
\end{equation}
which implies that $V_0 = V_1$ and $S_1 = (\id - S_0^2)^{1/2}$. In summary, a general two-outcome measurement can be performed by applying a unitary $V_0$, followed by a two-outcome measurement with diagonal Kraus operators $S_0$ and $S_1$, and -- depending on the measurement outcome -- completed by a conditional unitary $U_0$ or $U_1$. A setup realizing this procedure on photon polarization is shown in Fig.~\ref{fig:QubitMeasurement}. \begin{figure}
	\centering
	\includegraphics[scale=0.07]{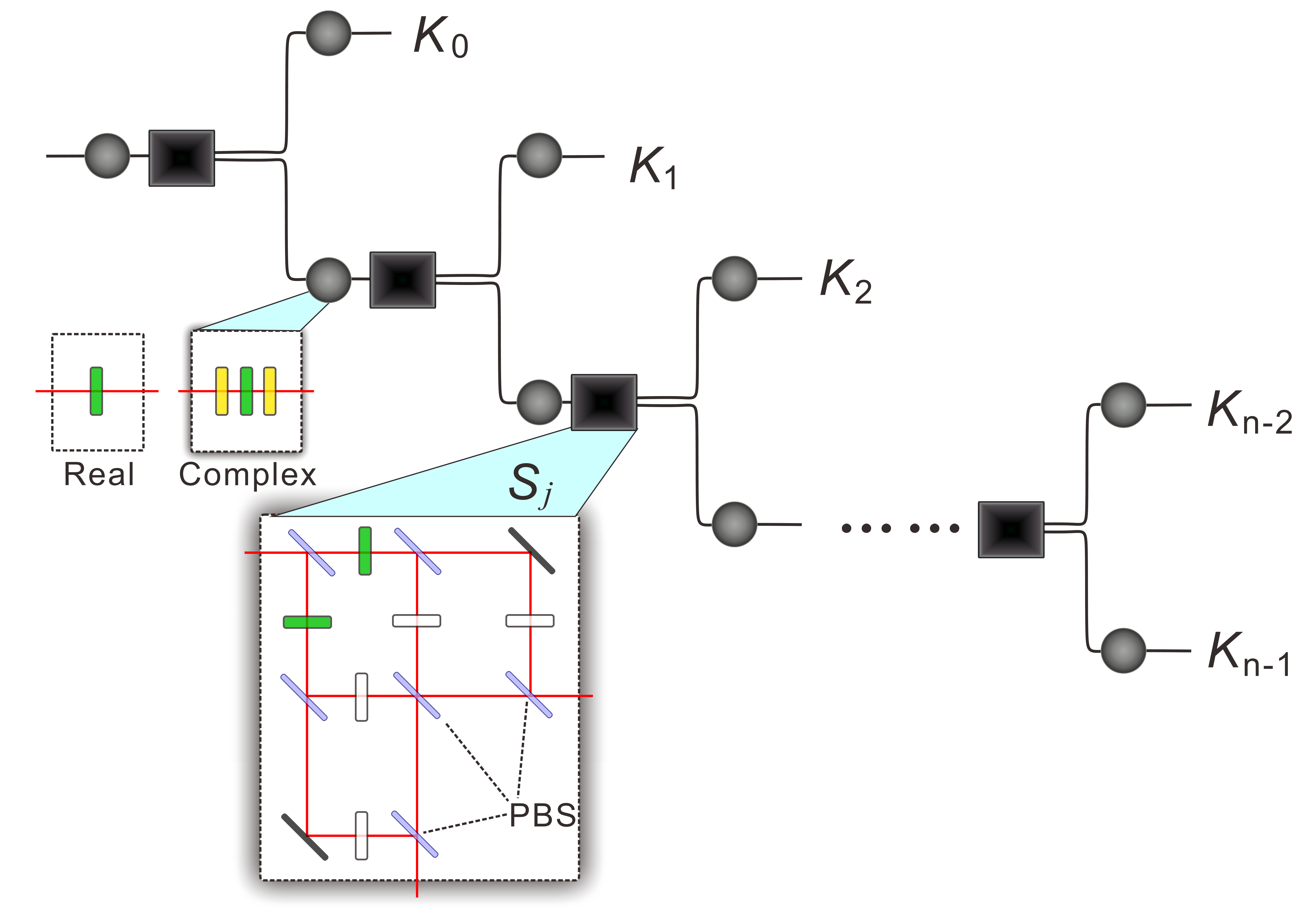}
	\caption{\label{fig:QubitMeasurement}
		\textbf{Linear optical implementation of real and general qubit operations with polarized photons.} Green and yellow plates represent unset wave plates (WP), while gray ones represent fixed wave plates that we don't need to control.  Thin strips denote beamsplitters which can separate the horizontally polarized photons from vertically polarized ones. Round nodes represent orthogonal or unitary operations, and boxes represent the operations $S_j$, which have two outcomes. For general real operations we only need to control half-wave plates, while for complex ones we have to add quarter-wave plates for manipulating the imaginary part of the photonic states.
	}
\end{figure}
The unitaries $V_0$, $U_0$, and $U_1$ on the polarization-encoded qubit can be realized by $3$ wave plates per unitary, while the measurement with diagonal Kraus operators $\{S_0,S_1\}$ can be realized with $3$ beam displacers and $5$ wave plates, of which 2 are unset. This amounts to $11$ unset wave plates in total. By using the same procedure repeatedly, this setup can be extended to $n$ Kraus operators, see also~\cite{AhnertPhysRevA.71.012330}. For each additional Kraus operator we need $8$ unset wave plates, giving $8n-5$ unset wave plates in total, as claimed. 

If all Kraus operators are real, fewer wave plates are needed. This can be seen from the fact that the singular value decomposition of each $K_j$ can be done with real $U_j$ and $V_j$. Thus, a real measurement with two outcomes can be implemented with $5$ unset wave plates, and each additional real Kraus operator requires $4$ additional wave plates, see also Fig.~\ref{fig:QubitMeasurement}. The number $4n-3$ is optimal, since it corresponds to the number of independent real parameters for $n$ real Kraus operators. Compared to $8n-5$ unset wave plates for a general $n$-outcome measurement via the method presented above, in the limit $n\rightarrow \infty$ we can save approximately half of the optical elements if we restrict ourselves to real measurements.

We will now go one step further and consider implementation of quantum operations of arbitrary dimension.
Note that every real operation acting on a system of dimension $d$ has a real dilation~\cite{Hickey+Gour.JPA.2018}:
\begin{equation}
\Lambda^A_\mathrm{RO}[\rho^A]=\tr_B \left[O_{AB}\left(\rho^A \otimes \ket{0}\!\bra{0}^B\right)O_{AB}^T \right], \label{eq:Dilation}
\end{equation}
where $O_{AB}$ is a $d^3 \times d^3$ real orthogonal matrix. Correspondingly, a general quantum operation admits a dilation with a general $d^3 \times d^3$ unitary matrix. Implementing an $m\times m$ unitary in path degree requires at least $m^2-1$ unset wave plates, corresponding to the number of real parameters of the unitary. On the other hand, an $m \times m$ real orthogonal matrix can be decomposed into $(m^2-m)/2$ real orthogonal matrices, each acting on two levels. This can be proven in the same way as for unitary matrices, see e.g.~\cite[p. 189]{NielsenChuang2011}. There, an explicit construction is presented for decomposing an $m\times m$ unitary matrix $U$ into $(m^2-m)/2$ two-level unitaries. If $U$ is additionally real, all two-level unitaries constructed in the proof are also real. Since a real orthogonal two-level matrix can be implemented with a single wave plate, any real orthogonal $m \times m$ matrix can be implemented by using $(m^2-m)/2$ unset wave plates. Thus, implementing a real operation can be achieved with $(d^6-d^3)/2$ unset wave plates. Instead, implementing a general quantum operation in the same way requires at least $d^6-1$ unset wave plates. For large $d$, restricting ourselves to real operations reduces the number of unset wave plates by $1/2$, when compared to the number of wave plates for a general quantum operations implemented via a unitary dilation. 

In summary, our results show that restricting ourselves to real operations in optical experiments allows to reduce the number of unfixed wave plates by $1/2$, in the limit of large system dimension. Similar results are found for single-qubit measurements with $n$ outcomes:  in the limit $n \rightarrow \infty$ restricting ourselves to real qubit measurements allows to reduce the number of unfixed wave plates by $1/2$. These results equip the resource theory of imaginarity with an operational meaning in optical experiments.

\section{Imaginarity in local state discrimination \label{sec:StateDiscrimination}}

We will now discuss the role of imaginarity for discrimination of quantum states. For two mixed states $\rho^{AB}_1$ and $\rho^{AB}_2$ to be perfectly distinguishable via LOCC, there must exist a POVM with elements $\{M_1, M_2\}$ of the form\footnote{Additionally to Eq.~(\ref{eq:SeparablePOVM}), the POVM should be implementable as LOCC.}
\begin{equation} \label{eq:SeparablePOVM}
    M_j = \sum_k A_{j,k} \otimes B_{j,k}
\end{equation}
with Hermitian $A_{j,k}$ and $B_{j,k}$, and moreover 
\begin{subequations} \label{eq:StateDiscriminationCondition}
\begin{align}
    \tr \left[M_1\rho^{AB}_1\right] &= \tr \left[M_2\rho^{AB}_2\right] = 1, \\
    \tr \left[M_1\rho^{AB}_2\right] &= \tr \left[M_2\rho^{AB}_1\right] = 0.
\end{align}
\end{subequations}
As was shown in~\cite{walgate2000local} such perfect discrimination is indeed possible if $\rho_1^{AB}$ and $\rho_2^{AB}$ are pure and orthogonal.

Here, we will consider state discrimination via \emph{local real operations and
classical communication} (LRCC)~\cite{PRLversion}. The set LRCC is defined in the same way as LOCC, but with the constraint that both Alice and Bob perform only real operations locally. Then, for two states $\rho^{AB}_1$ and $\rho^{AB}_2$ to be perfectly distinguishable via LRCC there must exist a POVM fulfilling Eqs.~(\ref{eq:StateDiscriminationCondition}), with POVM elements of the form~(\ref{eq:SeparablePOVM}) and real symmetric $A_{j,k}$ and $B_{j,k}$. 

As we will see in the following, perfect discrimination via LRCC is possible for any pair of pure orthogonal real states.

\begin{proposition}
Two real orthogonal pure quantum states $\ket{\psi}^{AB}$ and $\ket{\phi}^{AB}$ can be perfectly distinguished via local real operations and classical communication.
\end{proposition}
\begin{proof}
The proof follows similar lines of reasoning as in Ref.~\cite{walgate2000local}. Any two real pure states can be expanded as
\begin{subequations}
\begin{align}
    & \ket{\psi}^{AB}=\sum_{j=0}^{d-1}\ket{j}\ket{a_j},\\
    & \ket{\phi}^{AB}=\sum_{j=0}^{d-1}\ket{j}\ket{b_j},
\end{align}
\end{subequations}
where $\ket{a_j}$ and $\ket{b_j}$ are (unnormalized) real states and $d=d_A$ is Alice's dimension. Without loss of generality we assume that $d_A \leq d_B$, where $d_B$ is Bob's dimension.

We now consider the matrix $C$ with elements
\begin{equation}
    C_{jk}=\langle a_j | b_k \rangle.
\end{equation}
Since $\ket{\psi}^{AB}$ and $\ket{\phi}^{AB}$ are orthogonal, we have 
\begin{equation}
    \tr C=0.
\end{equation}
If we apply a real orthogonal matrix $O$ on Alice's side, the two states are transformed as
\begin{subequations}
\begin{align}
    & (O\otimes\id)\ket{\psi}^{AB}=\sum_{k=0}^{d-1}\ket{k}\sum_{j=0}^{d-1}O_{kj}\ket{a_j},\\
    & (O\otimes\id)\ket{\phi}^{AB}=\sum_{k=0}^{d-1}\ket{k}\sum_{j=0}^{d-1}O_{kj}\ket{b_j}.
\end{align}
\end{subequations}
If now Alice applies a local von Neumann measurement in the computational basis, Bob is left with a (possibly unnormalized) state of the form
\begin{equation}
    \ket{\tilde{a}_k} = \sum_j O_{kj}\ket{a_j} \,\,\,\,\mathrm{or}\,\,\,\,\ket{\tilde{b}_k} = \sum_j O_{kj}\ket{b_j}.
\end{equation}
This allows us to define a matrix $\tilde{C}$ as follows:
\begin{equation}
    \tilde{C}_{mn}=\braket{\tilde{a}_m}{\tilde{b}_n}=\sum_{kl}O_{mk}\langle a_k | b_l \rangle O_{nl}= \sum_{kl}O_{mk} C_{kl} (O^T)_{ln},
\end{equation}
so we have $\tilde{C}=O C O^T$. 

In the next step we will show that there exists a real orthogonal matrix $O$ such that all diagonal elements of $\tilde{C}$ become zero. This will complete the proof: if Alice applies $O$ locally and performs a von Neumann measurement in the computational basis, Bob will find his system either in the state $\ket{\tilde{a}_j}$ or $\ket{\tilde{b}_j}$. Bob can distinguish these states perfectly, since 
\begin{equation}
    \tilde{C}_{jj} = \braket{\tilde{a}_j}{\tilde{b}_j}=0.
\end{equation}

Note that for any $2\times 2$ real matrix $C$, there always exists a real orthogonal $2\times2$ matrix $O$ such that the diagonal elements of $OCO^T$ are equal to each other. Assume now that the dimension of Alice is a power of $2$, i.e., $d_A=2^k$. This implies that $C$ is a $2^k \times 2^k$ real matrix. Our goal is to make all the diagonal elements of $C$ zero by applying two-level real orthogonal rotations. Recalling that the trace of $C$ is zero, our goal can be achieved by making all the diagonal elements equal.

To this end, we first group all diagonal elements of $C$ into $2^{k-1}$ pairs and apply $2^{k-1}$ real orthogonal transformations, each acting on two levels. In this way we can obtain a new matrix $C'$ with the property $\tr C' = \tr C$ and the diagonal elements of $C'$ are pairwise equal. Consider now two pairs of diagonal elements, e.g.
\begin{equation}
C'_{00}=C'_{11},\,\,\,\,\,\, C'_{22}=C'_{33}.
\end{equation}
We can now apply two real orthogonal transformations, one acting on levels $0$ and $2$, and the other acting on levels $1$ and $3$. In this way, with a suitable choice of real orthogonal transformations, we can obtain a new matrix $C''$ with the properties $\tr C'' = \tr C$ and
\begin{equation}
    C''_{00}=C''_{11}=C''_{22}=C''_{33}.
\end{equation}
Proceeding in this way, we can make all diagonal elements equal to zero. This completes the proof for the case that the dimension of Alice's system is a power of $2$. 

If the dimension of Alice's system is not a power of $2$, we can extend the dimension of Alice to be of the form $2^k$, thus extending the correlation matrix $C$ with additional rows and columns having zero in all entries. All parts of the proof remain the same, which proves the statement for any dimension of Alice.
\end{proof}

As we discuss in~\cite{PRLversion}, the situation is very different when mixed states are considered: there exist pairs of mixed real states $\rho^{AB}_1$ and $\rho^{AB}_2$ which can be distinguished perfectly with LOCC, but which cannot be distinguished via LRCC with any nonzero probability. We refer to~\cite{PRLversion} for more details, where we also report results on experimental state discrimination with linear optics.

\section{Conclusion \label{sec:Conclusion}}
In this work we investigated the role of complex numbers in quantum mechanics, using the framework of quantum resource theories. We discussed imaginarity quantification, focusing on the geometric imaginarity and the robustness of imaginarity, proving that both measures have an operational meaning for state conversion. We presented a full solution for stochastic state conversion via real operation for pure states, and for deterministic state conversion for all states of a single qubit. We also found optimal fidelity for approximate imaginarity distillation in the single-copy regime. 

Our results show that imaginarity can be regarded as a resource in optical experiments: under certain assumptions commonly used in experiments, realizing a real operation reduces the number of optical elements by $1/2$, compared to the number of elements for a general quantum operation. We also discuss the role of imaginarity for local discrimination of quantum states. The methods presented here are also used in the companion article~\cite{PRLversion}, where several of the results have been initially announced. Our work can stimulate new research on quantum resource theories, and is important for a deeper understanding of the role of complex numbers in quantum mechanics.

\section*{Acknowledgements}
We thank Bartosz Regula and Micha\l{} Oszmaniec for discussion. The work at the University of Science and Technology of China is supported by the National Key Research and Development Program of China (No. 2018YFA0306400), the National Natural Science Foundation of China (Grants No. 61905234, 11974335,11574291, and 11774334), the Key Research Program of
Frontier Sciences, CAS (Grant No. QYZDYSSW-SLH003), and the Fundamental Research Funds for the Central Universities (Grant No. WK2470000026). T.V.K., S.R., and A.S.\ acknowledge financial support by the ``Quantum Optical Technologies'' project, carried out within the International Research Agendas programme of the Foundation for Polish Science co-financed by the European Union under the European Regional Development Fund. C.M.S.\ acknowledges the hospitality of the Centre for Quantum Optical Technologies at the University of Warsaw, and financial support by the Pacific Institute for the Mathematical Sciences (PIMS) and a Faculty of Science Grand Challenge award at the University of Calgary.

\appendix

\bibliography{ImaginarityBib}

\end{document}